\documentclass[final]{siamltex}

\usepackage{amsmath}
\usepackage{amssymb}
\usepackage{latexsym}
\usepackage{amscd}
\usepackage{enumerate}

\title{Hankel determinant solutions to several discrete integrable systems and the Laurent Property}

\author{Xiang-Ke Chang\footnotemark[1] \footnotemark[2],
  Xing-Biao Hu\footnotemark[1], Guoce Xin\footnotemark[3]}

\begin{document}
\maketitle

\renewcommand{\thefootnote}{\fnsymbol{footnote}}
\footnotetext{Emails: changxk@lsec.cc.ac.cn; hxb@lsec.cc.ac.cn; guoce.xin@gmail.com}
\footnotetext[1]{LSEC, Institute of Computational Mathematics and
  Scientific Engineering Computing, AMSS, Chinese Academy of Sciences,
  P.O.Box 2719, Beijing 100190, PR China.}
\footnotetext[2]{University of Chinese Academy of Sciences,\
  Beijing, PR China.}
\footnotetext[3]{Department of mathematics, Capital Normal University,
Beijing 100048, PR China. }

\begin{abstract}
Many discrete integrable systems exhibit the Laurent phenomenon. In this paper, we investigate three integrable systems: the Somos-4 recurrence, the Somos-5 recurrence and a system  related to so-called $A_1$ $Q$-system, whose general solutions are derived in terms of Hankel determinant. As a result, we directly confirm that they satisfy the Laurent property. Additionally, it is shown that the Somos-5 recurrence can be viewed as a specified B\"{a}cklund transformation of the Somos-4 recurrence. The related topics about Somos polynomials are also studied.
\end{abstract}



\pagestyle{myheadings}
\thispagestyle{plain}
\markboth{TEX PRODUCTION}{Determinant Solutions and the Laurent Property}


\section{Introduction}
Laurent phenomenon is a crucial property behind integrality shared by a class of combinatorial models while integrability is a key feature for a class of what we call integrable systems. An interesting observation is that many discrete integrable systems exhibit the Laurent phenomenon, and many mappings with the Laurent property are proved to be integrable. In the recent decade, a lot of related work were done via the study of a class of commutative algebras, called cluster algebras \cite{fomin2002cluster,fomin2007cluster}, for its success in proving the Laurent phenomenon. For instance, see \cite{di2010solution,di2011discrete,di2013t,fomin2002laurent,fordy2011mutation,fordy2011symplectic,fordy2012discrete,fordy2011cluster,gekhtman2011generalized,hone2006diophantine,hone2007singularity}.
More recently, a generalization of cluster algebras called Laurent phenomenon algebras \cite{lam2012laurent} was also introduced in order to study the Laurent phenomenon \cite{alman2013laurent,hone2014family}.

One question is whether we can prove the Laurent property (or a stronger property) from other views. In this paper, we shall give one possible choice---the view of explicit determinant solution. That is, the Laurent property of some discrete integrable systems can be proved by their determinant solutions. Of course, how to obtain the desired determinant formulae is another challenging problem. This approach succeeds for the following three integrable systems:

The first two systems are the Somos-4 recurrence and the Somos-5 recurrence defined by
\begin{equation}\label{somos4}
S_{n}S_{n-4}=\alpha S_{n-1}S_{n-3}+\beta S_{n-2}^2,
\end{equation}
and
\begin{equation}\label{somos5}
S_{n}S_{n-5}=\tilde{\alpha} S_{n-1}S_{n-4}+\tilde{\beta} S_{n-2}S_{n-3},
\end{equation}
with nonzero parameters $\alpha, \beta, \tilde{\alpha}, \tilde{\beta}$ and arbitrary initial values, respectively. Here we assume $S_i$ never vanish. By using their theory of cluster algebras \cite{fomin2002cluster,fomin2007cluster}, Fomin and Zelevinsky \cite{fomin2002laurent,hone2007laurent,hone2008integrality} proved that the Somos-4 and Somos-5 recurrences exhibit the Laurent property, that is, $S_n$ are Laurent polynomials in its initial values with coefficients in $\mathbb{Z}[\alpha, \beta]$ and $\mathbb{Z}[\tilde\alpha,\tilde\beta]$, respectively. Hone constructed explicit solutions to the Somos-4 and Somos-5 recurrences in terms of the Weierstrass sigma function in \cite{hone2005elliptic} and \cite{hone2007sigma}, respectively. He also indicated that the Somos-4 recurrence can be thought of as an integrable symplectic mapping \cite{bruschi1991integrable,veselov1991integrable} by making a change of variables and Somos-5 be a integrable Poisson mapping. Unfortunately, it is not obvious to see the Laurent phenomenon from their Weierstrass sigma function formulae.

Our result for the Somos-4 recurrence with initial values $1,1,x,y$ is as follows.
\begin{theorem}\label{theorem_somos4}
  If we let $S_{-1}=1,S_0=1$,$S_1=x,S_2=y,S_3=\alpha y+\beta x^2$ by shifting the indices, then the Somos-4 recurrence \eqref{somos4} has the following explicit Hankel determinant representation:
\[
S_n=\det(p_{i+j})_{0\leq i,j\leq n-1},
\]
where $p_m=\frac{\beta x^2-y^2+\alpha x^3-\alpha y}{\sqrt\alpha xy}p_{m-1}+\frac{\beta+\alpha x-xy}{y}p_{m-2}+\sum_{k=0}^{m-2}p_kp_{m-2-k}(m\geq2)$ with initial values $p_0=x,p_1=-\sqrt\alpha.$
  \end{theorem}

The Laurent phenomenon clearly holds for initial values $1,1,x,y$ by the theorem: i) the determinant formula asserts that $S_n$ is a Laurent polynomial in $\sqrt{\alpha}, x,y$ and polynomial in $\beta$; ii) the recursion \eqref{somos4} shows that $S_n$ is rational in $\alpha$ and that $S_n$ exists and not equal to $0$ when $\alpha=0$, which can be checked by induction.
To see that the Laurent phenomenon holds for arbitrary initial values, we observe that the recursion \eqref{somos4} holds for $S_n$ if and only if it holds for $S_n/c$ for constant $c\ne 0$, if and only if it holds for $S_n t^n$ for constant $t\ne 0$. This makes it sufficient to consider the sequence $\{\frac{S_n}{S_0}(\frac{S_0}{S_1})^n \}_{n\ge 0}$, which starts with $1,1$. Note that the reduction to the case $1,1,x,y$ is a simple but important step, since
the Somos-4 recurrence with arbitrary initial values does not seem to have the desired Hankel determinant representation.

Similarly, it is sufficient to give the following result for the Somos-5 recurrence.
\begin{theorem}\label{theorem_somos5}
 If we let $S_{-2}=1, S_{-1}=1, S_0=x$,$S_1=y,S_2=z,S_3=\tilde{\alpha}z+\tilde{\beta}xy,$
 $S_4=\tilde{\alpha}^2xz+\tilde{\alpha}\tilde{\beta}x^2y+\tilde{\beta}yz$ by shifting the indices, then the Somos-5 recurrence \eqref{somos5} has the following explicit Hankel determinant representation:

\begin{eqnarray}
S_{2n}&=&x^{n+1}\det(p_{i+j})_{0\leq i,j\leq n-1},\\
S_{2n+1}&=&y^{n+1}\det(q_{i+j})_{0\leq i,j\leq n-1},
\end{eqnarray}
where $p_m$ is recursively given by $p_0=\frac{z}{x^2},p_1=-\tilde{\beta}$,
 \begin{multline*}
p_m=-\frac{-\tilde{\alpha}x^2z+\tilde{\alpha}x^2y^2+\tilde{\beta}x^3y+zy^2-z^2}{xyz}p_{m-1}\\
+\frac{yx^4\tilde{\alpha}^2+x^5\tilde{\alpha}\tilde{\beta}+\tilde{\beta}zx^3-z^2y}{x^2yz}p_{m-2}+
\sum_{k=0}^{m-2}p_kp_{m-2-k}(m\geq2),
\end{multline*}
and $q_m$ is recursively given by $q_0=\frac{\tilde{\alpha}z+\tilde{\beta} xy}{y^2},q_1=-\tilde{\beta}$, \begin{multline*}
q_m=\frac{\tilde{\alpha}x^2z+\tilde{\alpha}x^2y^2-z^2-zy^2+\tilde{\beta}x^3y}{xyz}q_{m-1}\\
+\frac{\tilde{\alpha}y^4-\tilde{\alpha}z^2+\tilde{\beta}xy^3-\tilde{\beta}xzy}{y^2z}q_{m-2}+\sum_{k=0}^{m-2}q_kq_{m-2-k}(m\geq2).
\end{multline*}
  \end{theorem}

 The third system is called the extended $A_1$ $Q$-system:
\begin{equation}\label{new_recurrence}
S_{n}S_{n-2}=S_{n-1}^2+\beta,
\end{equation}
with nonzero parameters $\beta$ and arbitrary initial values. Here we assume $S_i$ never vanish. The special case when $\beta=1$ and all the initial values are 1s reduces to the $A_1$ $Q$-system \cite{di2011discrete,di2010qsystem,di2011non}. From the view of integrable systems, \eqref{new_recurrence} is C-integrable\cite{calogero1991certain}.  Note that the Laurent property of \eqref{new_recurrence} was shown in \cite{di2011discrete,fomin2002laurent,hone2007singularity}.

By dividing \eqref{new_recurrence} by $S_0^2$, we may assume the initial value is $S_0=1, S_1=x$ ($\beta$ changes but do not affect the Laurent property). The determinant solution is as follows, from which the Laurent property is obviously satisfied.
\begin{theorem} \label{theorem_A1Q}
  If we let $S_0=1$,$S_1=x$, then the extended $A_1$ $Q$-system \eqref{new_recurrence} has the following explicit Hankel determinant representation:
\[
S_n=\det(p_{i+j})_{0\leq i,j\leq n-1},
\]
where $p_m=\frac{\beta+1-x^2}{x}p_{m-1}+\sum_{k=0}^{m-1}p_kp_{m-1-k}(m\geq2)$ with initial values $p_0=x,p_1=1.$
  \end{theorem}

All three theorems can be proved by employing Sulanke and Xin's method of quadratic transformation for Hankel determinants \cite{sulanke2008hankel,xin2009proof} or classic determinant technique. In particular, the special case $x=y=\alpha=\beta=1$ of Theorem \ref{theorem_somos4} was obtained by the third named author in \cite{xin2009proof}. It is also confirmed by using the classic determinant technique in \cite{chang2012conjecture}, which is used to solve a conjecture that a certain determinant satisfies a specified Somos-4 recurrence \cite{barry2010generalized,barry2011invariant}. In Section 3, we shall give a detailed proof for the Somos-4 case by using Sulanke and Xin's method of quadratic transformation, and prove the result for the extended $A_1$ $Q$-system by applying the classic determinant technique. As for the Somos-5 recurrence, we can solve it via the B\"{a}cklund transformation \cite{hirota2004direct,rogers1982backlund} of the Somos-4 recurrence, which will be introduced in Section 2. As a result, the Laurent property for the three models are confirmed from their determinant solutions. In Section 4, we prove the Somos-4 polynomials result in \cite{website:somospoly} in our own way and give more Somos polynomial sequences.

\section{Preliminaries}
\subsection{Sulanke and Xin's quadratic transformation for Hankel determinants}

There are many classical tools for evaluating Hankel determinants, such as  orthogonal polynomials approach (cf. e.g.\cite{krattenthaler2005advanced,cigler2011some,viennot1983}),
Gessel-Viennot-Lindstr\"{o}m theorem (cf. ep.g. \cite{bressoud1999proofs,gessel1985binomial,sulanke2008hankel}),the J -fractions (cf. e.g. \cite{krattenthaler2005advanced,wall1948analytic}) and the S-fractions (cf. e.g.\cite{jones1980continued}). Recently, Sulanke and Xin proposed a method of quadratic transformation for Hankel determinants \cite{sulanke2008hankel}
developed from the continued fraction method of Gessel and Xin \cite{gessel2006generating}. In \cite{xin2009proof}, Xin applied the special quadratic transformation to solve Somos' conjecture about the Hankel determinant solution to the Somos-4 recurrence with initial values $S_i=1, i=0,1,2,3$ and $\alpha=1,\beta=1$. We restate the iterative transformation as follows.
\begin{lemma}\label{lemma_xin}
 Given the initial values $a_0,b_0,c_0,d_0,e_0,f_0$, let the generating function $Q_{0}(x)$ be the unique power series solution of
 \begin{eqnarray*}
   &&Q_{0}(x)=\frac{a_{0}+b_{0}x}{1+c_{0}x+d_{0}x^2+x^2(e_{0}+f_{0}x)Q_{0}(x)}.
 \end{eqnarray*}
 If $a_{n+1},b_{n+1},c_{n+1},d_{n+1},e_{n+1},f_{n+1}$ are recursively defined by
 \begin{eqnarray*}
   &&a_{n+1}=-\frac{a_n^3e_n+a_n^2d_n-a_nb_nc_n+b_n^2}{a_n^2},\\
   &&b_{n+1}=-\frac{a_n^4f_n+c_na_n^3d_n-a_n^2c_n^2b_n+2a_nc_nb_n^2-a_n^2b_nd_n-b_n^3}{a_n^3},\\
   &&c_{n+1}=c_n,\\
   &&d_{n+1}=-\frac{a_n^2d_n-2a_nb_nc_n+2b_n^2}{a_n^2},\\
   &&e_{n+1}=-1,\\
   &&f_{n+1}=-\frac{b_n}{a_n},
 \end{eqnarray*}
 where we suppose that $a_n\neq0,n=0,1,2,\cdots$, then
 \begin{equation*}
   \det(H_n(Q_0))=a_0^{n}a_1^{n-1}\cdots a_{n-1}.
 \end{equation*}
  Here we remark that $H_n(R)$ denotes the Hankel determinant $\det(r_{i+j})_{0\leq i,j\leq n-1}$ for any power series $R(x)=\sum_{n=0}^\infty r_nx^n$.
\end{lemma}

The basic tool in our approach is the following surprising result, which solves the above recursion system.
 \begin{lemma}[Theorem 2 in \cite{xin2009proof}]\label{theorem_xin}
   Suppose $c_n=c,e_n=-1$, then the recursion system in Lemma \ref{lemma_xin} satisfies
   \begin{equation}
     a_{{n+2}}a_{{n+1}}+a_{{n+1}}a_{{n}}=2a_{{0}}a_{{1}}+a_{{0}} \left( 2\,f_{{1}}+c \right)  \left( f_{{0}}+c+f_{{1}} \right)-{\frac {{a_{{0}}}^{2} \left( f_{{0}}+c+f_{{1}} \right) ^{2}}{a_{{n+1}}}}.
   \end{equation}
 \end{lemma}
\subsection{B\"{a}cklund transformation}

A B\"{a}cklund transformation is a transformation between a solution $u$ of a given
differential or difference equation,
 \begin{equation*}
  \mathcal{L}_1(u)=0,
 \end{equation*}
and another solution $v$ of another differential or difference equation,
 \begin{equation*}
  \mathcal{L}_2(v)=0,
 \end{equation*}
where $\mathcal{L}_2$ may be the same as, or different from, $\mathcal{L}_1$. It is an important tool for finding new solutions in soliton theory and integrable systems. See \cite{hirota2004direct,rogers1982backlund} for more details. Here we shall give an example, which is concerned about B\"{a}cklund transformation in bilinear form between two solutions of the Somos-4 recurrence. It can be used for solving the Somos-5 recurrence.

Firstly, we rewrite the Somos-4 recurrence as a bilinear form
\begin{equation}\label{bi_somos4}
  (e^{2D_{n}}-\alpha e^{D_{n}}-\beta)S_n\cdot S_n=0,
\end{equation}
where  the Hirota's bilinear operators \cite{hirota2004direct} are
defined as follows:
\begin{eqnarray*}
&&D_z^mD_t^ka\cdot b\equiv(\frac{\partial}{\partial z}-\frac{\partial}{\partial z^\prime})^m(\frac{\partial}{\partial t}-\frac{\partial}{\partial t^\prime})^ka(z,t)b(z^\prime,t^\prime)|_{z^\prime=z,t^\prime=t},\\
&&e^{\delta D_n}a_n\cdot b_n\equiv \exp\
[\delta(\frac{\partial}{\partial n}-\frac{\partial}{\partial
n^\prime})]a_nb_{n^\prime}|_{n^\prime=n}=a_{n+\delta}b_{n-\delta}.
\end{eqnarray*}

Then, we have the following lemma.
\begin{lemma}\label{lemma_somosbt}
The Somos-4 recurrence has the following bilinear B\"{a}cklund transformation:
\begin{eqnarray}
&&(e^{-D_{n}}-\lambda e^{D_{n}}-\mu )f_n\cdot g_n=0,\label{somos4_backlund_1}\\
&&(\mu e^{3/2D_{n}}-\eta e^{-1/2D_{n}}-\alpha e^{1/2D_{n}})f_n\cdot g_n=0,\label{somos4_backlund_2}
\end{eqnarray}
where $\mu,\lambda,\eta$ are arbitrary constants.
\end{lemma}
\begin{proof}
  Let $f_n$ be a solution of the Somos-4 recurrence. What we need to prove is that $g_n$ satisfying
\eqref{somos4_backlund_1} and \eqref{somos4_backlund_2} is another solution of the Somos-4 recurrence, i.e.,
\begin{equation*}
  P\equiv(e^{2D_{n}}-\alpha e^{D_{n}}-\beta)g_n\cdot g_n=0,
\end{equation*}
Then, using the bilinear operator identities
\begin{eqnarray*}
  &&(e^{2D_{n}}f_n\cdot f_n)g_n^2-f_n^2(e^{2D_{n}}g_n\cdot g_n)=2\sinh(D_n)(e^{D_{n}}f_n\cdot g_n)\cdot(e^{-D_{n}}f_n\cdot g_n),\\
  &&(e^{D_{n}}f_n\cdot f_n)g_n^2-f_n^2(e^{D_{n}}g_n\cdot g_n)=2\sinh(1/2D_n)(e^{1/2D_{n}}f_n\cdot g_n)\cdot(e^{-1/2D_{n}}f_n\cdot g_n),\\
  &&\sinh(D_n)(e^{D_{n}}f_n\cdot g_n)\cdot(f_n\cdot g_n)=\sinh(1/2D_n)(e^{3/2D_{n}}f_n\cdot g_n)\cdot(e^{-1/2D_{n}}f_n\cdot g_n),\\
  &&\sinh(\delta D_n)f_n\cdot f_n=0,
\end{eqnarray*}
we have
\begin{eqnarray*}
  -Pf_n^2&=&[(e^{2D_{n}}-\alpha e^{D_{n}}-\beta)f_n\cdot f_n]g_n^2-f_n^2(e^{2D_{n}}-\alpha e^{D_{n}}-\beta)g_n\cdot g_n\\
  &=&[(e^{2D_{n}}f_n\cdot f_n)g_n^2-f_n^2(e^{2D_{n}}g_n\cdot g_n)]-\alpha [(e^{D_{n}}f_n\cdot f_n)g_n^2-f_n^2(e^{D_{n}}g_n\cdot g_n)]\\
  &=&2\sinh(D_n)(e^{D_{n}}f_n\cdot g_n)\cdot(e^{-D_{n}}f_n\cdot g_n)\\
  &&-2\alpha\sinh(1/2D_n)(e^{1/2D_{n}}f_n\cdot g_n)\cdot(e^{-1/2D_{n}}f_n\cdot g_n)\\
  &=&2\mu\sinh(D_n)(e^{D_{n}}f_n\cdot g_n)\cdot(f_n\cdot g_n)\\
  &&-2\mu\sinh(1/2D_n)(e^{3/2D_{n}}f_n\cdot g_n)\cdot(e^{-1/2D_{n}}f_n\cdot g_n)\\
  &=&0.
\end{eqnarray*}
Thus the proof is completed.
\end{proof}
\section{Determinant solution and Laurent property}

\subsection{On the Somos-4 recursion system}

The Hankel determinant formula in Theorem \ref{theorem_somos4} is discovered by employing Sulanke and Xin's quadratic transformation for Hankel determinants. The proof is given below. 

{\em Proof of Theorem \ref{theorem_somos4}.}
Let $P(t)=\sum_{n=0}^\infty p_nt^n$ be the generating function of the sequence $\{p_n\}_{n=0}^\infty$, we only need to prove $S_n=\det(H_n(P))$. Additionally, it is easy to see that $P(t)$ satisfies
\begin{equation}
  P(t)=\frac{a_{0}+b_{0}t}{1+c_{0}t+d_{0}t^2+t^2(e_{0}+f_{0}t)P(t)}
\end{equation}
where
 \begin{eqnarray*}
   &&a_{0}=x,\\
   &&b_{0}=-\frac{\beta x^2-y^2+\alpha x^3}{\sqrt\alpha y},\\
   &&c_{0}=-\frac{\beta x^2-y^2+\alpha x^3-\alpha y}{\sqrt\alpha xy},\\
   &&d_{0}=-\frac{\beta+\alpha x-xy}{y},\\
   &&e_{0}=-1,\\
   &&f_{0}=0.
 \end{eqnarray*}
By Lemma \ref{lemma_xin}, the recursion \eqref{somos4} we need to prove is transformed to
 \begin{equation}\label{a_somos4}
 a_na_{n-1}a_{n-2}=\alpha+\beta/a_{n-1}.
\end{equation}

Lemma \ref{lemma_xin} also gives $a_1=y/x^2$ and $f_1=\frac{\beta x^2-y^2+\alpha x^3}{\sqrt\alpha xy}$. Thus applying Lemma \ref{theorem_xin} gives
\begin{equation}\label{a_id}
  a_{n+2}=(\frac{\beta x^2+y^2+\alpha x^3+\alpha y}{xy})/a_{n+1}-\alpha/a_{n+1}^2-a_n.
\end{equation}

By substituting \eqref{a_id} with $n$ replaced by $n-2$ into \eqref{a_somos4}, we need to show that
\begin{equation} \label{a_id2}
 T(n):= (\frac{\beta x^2+y^2+\alpha x^3+\alpha y}{xy})a_{n-1}a_{n-2}-\alpha a_{n-2}-a_{n-1}^2a_{n-2}^2-\alpha a_{n-1}-\beta=0
\end{equation}
holds for $n\ge 2$.

We prove this by induction. It is easy to confirm $T(2)=0$. Assume $T(n-1)=0$. Replacing $a_{n-1}$ by \eqref{a_id} with $n$ replaced by $n-3$ in $T(n)$, we have
\begin{equation*}
  T(n)=  (\frac{\beta x^2+y^2+\alpha x^3+\alpha y}{xy})a_{n-2}a_{n-3}-\alpha a_{n-3}-a_{n-2}^2a_{n-3}^2-\alpha a_{n-2}-\beta=T(n-1)=0.
\end{equation*}
This completes the proof.
\endproof

\subsection{On the Somos-5 recursion system}

In finding the determinant solution of the Somos-5 recursion, the key observation is that the even and odd terms of the
Somos-5 sequence are both Somos-4 sequence. This observation has been made by Hone (Proposition 2.8 in \cite{hone2007sigma}), through
making the connection with a second order nonlinear mapping with a first integral. Here we give an alternative explanation from the view of B\"acklund transformation, which seems simpler and more intuitive.
Now we prove Theorem \ref{theorem_somos5} by using Lemma \ref{lemma_somosbt} and Theorem \ref{theorem_somos4}.

{\em Proof of Theorem \ref{theorem_somos5}.}
If we let $f_n=S_{2n}$ and $g_n=S_{2n+1}$, then we get two coupled equations

\begin{eqnarray}
&&g_{n+1}f_{n-1}-\tilde{\alpha}f_{n+1}g_{n-1}-\tilde{\beta}g_nf_n=0,\label{somos5_1}\\
&&f_{n+2}g_{n-1}-\tilde{\alpha}g_{n+1}f_n-\tilde{\beta}f_{n+1}g_n=0.\label{somos5_2}
\end{eqnarray}

Recall that, as is indicated in Lemma \ref{lemma_somosbt}, the Somos-4 recurrence \eqref{somos4} has the following B\"{a}cklund transformation

\begin{eqnarray}
&&g_{n+1}f_{n-1}-\lambda f_{n+1}g_{n-1}-\mu g_nf_n=0,\label{somos4bt1}\\
&&\mu f_{n+2}g_{n-1}-\eta g_{n+1}f_n-\alpha f_{n+1}g_n=0,\label{somos4bt2}
\end{eqnarray}
where $\mu,\lambda,\eta$ are arbitrary constants. When $\mu=\tilde{\beta},\lambda=\tilde{\alpha},$ $\eta=\tilde{\alpha}\tilde{\beta},\alpha=\tilde{\beta}^2$, it is easy to check that \eqref{somos4bt1}, \eqref{somos4bt2} are the same as \eqref{somos5_1},\eqref{somos5_2}.  This means that subsequences $f_n=S_{2n}$ and $g_n=S_{2n+1}$ of the Somos-5 recurrence \eqref{somos5} may satisfy the same Somos-4 recurrence \eqref{somos4}  with $\alpha=\tilde{\beta}^2$ and $\beta=\frac{\tilde{\alpha}(\tilde{\alpha}^3yxz+\tilde{\alpha}^2\tilde{\beta}y^2x^2+\tilde{\alpha}\tilde{\beta}y^2z+\tilde{\alpha}\tilde{\beta}z^2+2\tilde{\alpha}^2zyx+\tilde{\alpha}^2\tilde{\beta}x^2z+\tilde{\beta}^2\tilde{\alpha}yx^3)}{xyz}$, where $\beta$ is determined by the first five values of $f_n$. In fact, by use of the meaning of B\"{a}cklund transformation, this can be confirmed  by induction.  Thus, it is easy to construct the general solution to the Somos-5 recurrence \eqref{somos5}. And one can complete the proof after some substitutions and calculations.
\endproof

It is natural to ask if we can give a more direct proof to the Somos-4 recursion of $S_{2n}$ (or $S_{2n-1}$). The answer is positive. We sketch the idea for $S_{2n}$ in the end of this subsection. It is easy to perform the detailed proof by Maple.

We can eliminate the odd $S_n$ to obtain the recursion
\begin{equation}\label{even_recurrence}
  {\frac {S_{{2\,n-6}}S_{{2\,n}}-{\tilde\alpha}^{2}S_{{2\,n-4}}S_{{2\,n-2}}}{
\tilde\beta\, \left( S_{{2\,n-6}}S_{{2\,n-2}}+\tilde\alpha\,{S_{{2\,n-4}}}^{2}
 \right) }}={\frac {\tilde\beta\, \left( S_{{2\,n-6}}S_{{2\,n-2}}+\tilde\alpha\,{S
_{{2\,n-4}}}^{2} \right) }{S_{{2\,n-8}}S_{{2\,n-2}}-{\tilde\alpha}^{2}S_{{2
\,n-6}}S_{{2\,n-4}}}},
\end{equation}
using which we can prove
$$
S_{2n}S_{2n-8}=\alpha S_{2n-2}S_{2n-6}+\beta S_{2n-4}^2
$$
by induction.

Eliminating $S_{2n}$ in the above equation by using \eqref{even_recurrence}, we are left to show the nullity of a relation, saying $target(S_{2n-2},S_{2n-4},S_{2n-6},S_{2n-8})$. Eliminating $S_{2n-2}$ in $target(S_{2n-2},S_{2n-4},S_{2n-6},S_{2n-8})$ by using \eqref{even_recurrence} with $n$ replaced by $n-1$, we obtain $$target(S_{2n-2},S_{2n-4},S_{2n-6},S_{2n-8})=\frac{S_{2n-4}^2S_{2n-10}}{S_{2n-6}^2S_{2n-8}}target(S_{2n-4},S_{2n-6},S_{2n-8},S_{2n-10}).$$
Then the proof is completed after checking $target(S_{6},S_{4},S_{2,}S_{0})=0$.

\subsection{On the  extended $A_1$ $Q$-system}
Let us first prove Theorem \ref{theorem_A1Q} by classic determinant technique. This approach gives no hint how we discovered the Hankel determinant formula, but this proof seems to be shorter than using the Sulanke and Xin's quadratic transformation method.

{\em Proof of Theorem \ref{theorem_A1Q}.}
Let $H_n^{(l)}$ denote $\det(p_{i+j+l})_{0\leq i,j\leq n-1}$ and we shall use the conventions
 \begin{equation}\label{convention}
 \begin{array}{l}
 H_0^{(l)}=1,\\
 H_n^{(l)}=0 \ \ \text{for $n<0$.}
 \end{array}
 \end{equation}

Firstly, we assert that there hold
\begin{eqnarray}
&&H_n^{(1)}=1, \quad n\geq0;\label{h_n_1}\\
&&H_n^{(0)}=xH_{n-1}^{(0)}+\beta H_{n-2}^{(2)},\quad n\geq1\label{h_n_0}.
\end{eqnarray}
Employing the Jacobi determinant identity \cite{aitken1959determinants,brualdi1983determinantal}, we get
\begin{equation}
H_n^{(0)}H_{n-2}^{(2)}=H_{n-1}^{(0)}H_{n-1}^{(2)}-(H_{n-1}^{(1)})^2, \quad n\geq2. \label{jacobiid}
\end{equation}
By use of (\ref{h_n_1}) and (\ref{h_n_0}), replace $H_{n-1}^{(1)}$, $H_{n-2}^{(2)}$ and $H_{n-1}^{(2)}$ in (\ref{jacobiid}), then we can obtain
\begin{eqnarray*}
H_{n+1}^{(0)}H_{n-1}^{(0)}=(H_{n}^{(0)})^2+\beta,\quad n\geq1.
\end{eqnarray*}
Noting that $S_n=H_n^{(0)}$, thus it suffices to confirm the formulae (\ref{h_n_1}) and (\ref{h_n_0}).

It is noted that $p_m$ also satisfy the following recurrence:
\begin{equation}
p_m=\frac{\beta+1}{x}p_{m-1}+\sum_{k=1}^{m-1}p_kp_{m-1-k}, \quad  m\geq2. \label{recurrence}
\end{equation}
We will prove (\ref{h_n_1}) and (\ref{h_n_0}) by employing row and column operations for the determinants and using this recurrence.

Let's consider $H_n^{(1)}$ firstly.

\textit{Step 1}: Subtract the $i$-th column multiplied by $p_{n-1-i}$ from the $n$-th
column for $i=1,2,\cdots, n-1$ and subtract the ($n-1$)-th column
multiplied by $\frac{\beta+1}{x}$ from the $n$-th column. Next, applying a similar
procedure to the $(n-1)$-th,$(n-2)$-th,$\cdots,2$-nd columns, by using the recursion relation (\ref{recurrence}), we have
\begin{eqnarray*}
H_{n}^{(1)}&=&\left|\begin{array}{cccc}
p_{1}&0&\ldots&0\\
p_{2}&p_1p_{1}&\ldots&p_1p_{n-1}\\
p_{3}&\displaystyle{\sum_{i=1}^2p_ip_{3-i}}&\ldots&\displaystyle{\sum_{i=1}^2p_ip_{n+1-i}}\\
\vdots&\vdots &\ddots & \vdots\\
p_{n}&\displaystyle{\sum_{i=1}^{n-1}p_ip_{n-i}}&\ldots&\displaystyle{\sum_{i=1}^{n-1}p_ip_{2n-2-i}}
  \end{array}\right|\nonumber\\
  &=&\left|\begin{array}{ccc}
p_1p_{1}&\ldots&p_1p_{n-1}\\
\displaystyle{\sum_{i=1}^2p_ip_{3-i}}&\ldots&\displaystyle{\sum_{i=1}^2p_ip_{n+1-i}}\\
\vdots &\ddots & \vdots\\
\displaystyle{\sum_{i=1}^{n-1}p_ip_{n-i}}&\ldots&\displaystyle{\sum_{i=1}^{n-1}p_ip_{2n-2-i}}
  \end{array}\right|.\nonumber\\
\end{eqnarray*}

\textit{Step 2}: Perform row operations for the above determinant.
For fixed $k=2,3,\cdots,n-1$, we subtract the $i$-th row multiplied by
$p_{k+1-i}/p_1$ for $i=1,\cdots,k-1$ from the $k-$th row. Then it follows that
(\ref{h_n_1}) holds.

Now let's turn to the proof of (\ref{h_n_0}). Obviously, the result holds for $n=1$. In the following we consider the case of $n>1$.

\textit{Step 1}: Subtract the $i$-th column multiplied by $p_{n-1-i}$ from the $n$-th
column for $i=2,3,\cdots, n-1$ and subtract the ($n-1$)-th column
multiplied by $\frac{\beta+1}{x}$ from the $n$-th column. Next, applying a similar
procedure to the $(n-1)$-th,$(n-2)$-th,$\cdots,2$-nd columns, by using the recursion relation (\ref{recurrence}), we have
\begin{eqnarray*}
H_{n}^{(0)}&=&\left|\begin{array}{ccccc}
p_{0}&-\beta&0&\ldots&0\\
p_{1}&p_1p_{0}&p_1p_{1}&\ldots&p_1p_{n-2}\\
p_{2}&\displaystyle{\sum_{i=1}^2p_ip_{2-i}}&\displaystyle{\sum_{i=1}^2p_ip_{3-i}}&\ldots&\displaystyle{\sum_{i=1}^2p_ip_{n-i}}\\
\vdots&\vdots&\vdots &\ddots & \vdots\\
p_{n-1}&\displaystyle{\sum_{i=1}^{n-1}p_ip_{n-1-i}}&\displaystyle{\sum_{i=1}^{n-1}p_ip_{n-i}}&\ldots&\displaystyle{\sum_{i=1}^{n-1}p_ip_{2n-3-i}}
  \end{array}\right|\nonumber\\
  &=&x\left|\begin{array}{cccc}
p_1p_{0}&p_1p_{1}&\ldots&p_1p_{n-2}\\
\displaystyle{\sum_{i=1}^2p_ip_{2-i}}&\displaystyle{\sum_{i=1}^2p_ip_{3-i}}&\ldots&\displaystyle{\sum_{i=1}^2p_ip_{n-i}}\\
\vdots&\vdots &\ddots & \vdots\\
\displaystyle{\sum_{i=1}^{n-1}p_ip_{n-1-i}}&\displaystyle{\sum_{i=1}^{n-1}p_ip_{n-i}}&\ldots&\displaystyle{\sum_{i=1}^{n-1}p_ip_{2n-3-i}}
  \end{array}\right|\\
  &&+\beta\left|\begin{array}{cccc}
p_1&p_1p_{1}&\ldots&p_1p_{n-2}\\
p_2&\displaystyle{\sum_{i=1}^2p_ip_{3-i}}&\ldots&\displaystyle{\sum_{i=1}^2p_ip_{n-i}}\\
\vdots&\vdots &\ddots & \vdots\\
p_{n-1}&\displaystyle{\sum_{i=1}^{n-1}p_ip_{n-i}}&\ldots&\displaystyle{\sum_{i=1}^{n-1}p_ip_{2n-3-i}}
  \end{array}\right|\nonumber\\
\end{eqnarray*}

\textit{Step 2}: Perform row operations for the above two determinants.
For fixed $k=2,3,\cdots,n-1$, we subtract the $i$-th row multiplied by
$p_{k+1-i}/p_1$ for $i=1,\cdots,k-1$ from the $k-$th row. Then it follows that
(\ref{h_n_0}) holds.

Therefore, we complete the proof.
\endproof

\textbf{Remark}: If we let $f_n=\gamma^{n/2}S_n$, then $f_n$ satisfy $f_{n}f_{n-2}=f_{n-1}^2+\beta\gamma^{n-1}$ with $f_0=1$ and $f_1=x\sqrt\gamma$. And we can see that $\{f_n\}$ gives the Fibonacci sequence when  $\gamma=-1$, $\beta=1$ and $x=-\sqrt{-1}$. It means that every Fibonacci number can be expressed as a Hankel determinant.

Moreover, it is noted that the extended $A_1$ $Q$-system satisfy a three term recurrence \cite{hone2007singularity}. In the following we use the relations of Hankel determinants to show it.

\begin{theorem}
  The extended $A_1$ $Q$-system with $S_0=1,S_1=x$ satisfy the following three term recurrence
  \begin{equation}
    S_n=\frac{x^2+1+\beta}{x}S_{n-1}-S_{n-2}.
  \end{equation}
\end{theorem}
\begin{proof}
We shall firstly prove the formula
  \begin{equation}\label{h_n_0-2}
     H_n^{(0)}=xH_{n-1}^{(2)}-H_{n-2}^{(2)},
  \end{equation}
where the notations are the same as that in the previous proof.
  It is easy to see that
$H_n^{(0)}=xH_{n-1}^{(2)}+\tilde{H}_n$, where
\begin{eqnarray*}
\tilde{H}_n=\left|\begin{array}{ccccc}
0&p_1&p_2&\cdots &p_{n-1}\\
p_1&p_2&p_3&\cdots &p_{n}\\
p_2&p_3&p_4&\cdots &p_{n+1}\\
\vdots&\vdots &\vdots&\ddots & \vdots\\
p_{n-1}&p_{n}&p_{n+1}&\cdots &p_{2n-1}
  \end{array}\right|.\nonumber\\
\end{eqnarray*}
Based on this observation, it suffices to prove that $\tilde{H}_n=-H_{n-2}^{(2)}$.

\textit{Step 1}: Subtract the $i$-th column multiplied by $p_{n-1-i}$ from the $n$-th
column for $i=1,2,\cdots, n-1$ and subtract the ($n-1$)-th column
multiplied by $\frac{\beta+1}{x}$ from the $n$-th column. Next, applying a similar
procedure to the $(n-1)$-th,$(n-2)$-th,$\cdots,3$-rd columns, by using the recursion relation (\ref{recurrence}), we have
\begin{eqnarray*}
\tilde{H}_n&=&\left|\begin{array}{ccccc}
0&p_{1}&0&\ldots&0\\
p_1&p_{2}&0&\ldots&0\\
p_2&p_{3}&p_1p_{2}&\ldots&p_1p_{n-1}\\
p_3&p_{4}&\displaystyle{\sum_{i=1}^2p_ip_{4-i}}&\ldots&\displaystyle{\sum_{i=1}^2p_ip_{n+1-i}}\\
\vdots&\vdots&\vdots &\ddots & \vdots\\
p_{n-1}&p_{n}&\displaystyle{\sum_{i=1}^{n-2}p_ip_{n-i}}&\ldots&\displaystyle{\sum_{i=1}^{n-2}p_ip_{2n-3-i}}
  \end{array}\right|\nonumber\\
  &=&-\left|\begin{array}{ccc}
p_1p_{2}&\ldots&p_1p_{n-1}\\
\displaystyle{\sum_{i=1}^2p_ip_{4-i}}&\ldots&\displaystyle{\sum_{i=1}^2p_ip_{n+1-i}}\\
\vdots &\ddots & \vdots\\
\displaystyle{\sum_{i=1}^{n-2}p_ip_{n-i}}&\ldots&\displaystyle{\sum_{i=1}^{n-2}p_ip_{2n-3-i}}
  \end{array}\right|.\nonumber\\
\end{eqnarray*}

\textit{Step 2}: Perform row operations for the above determinant.
For fixed $k=2,3,\cdots,n-2$, we subtract the $i$-th row multiplied by
$p_{k+1-i}/p_1$ for $i=1,\cdots,k-1$ from the $k-$th row. Then it follows that
$\tilde{H}_n=-H_{n-2}^{(2)}$.

Next, noting that $S_n=H_n^{(0)}$, the three term recurrence can be obtained by combining \eqref{h_n_0} and \eqref{h_n_0-2}.
\end{proof}

\textbf{Remark}: When $\beta=x^2-1$, $\{S_n\}$ satisfy the three term recurrence
  \begin{equation*}
    S_n=2xS_{n-1}-S_{n-2}.
  \end{equation*}
with initial data $S_0=1$ and $S_1=x$. This just meets the Chebyshev polynomials.

It is also noted that every sequence produced by the extended $A_1$ $Q$-system is a Somos-4 sequence \cite{swart2003elliptic}. More precisely, if $\{S_n\}$ satisfy \eqref{new_recurrence}, then $\{S_n\}$ also satisfy
\begin{equation*}
  S_{n+2}S_{n-2}=\frac{(x^2+1+\beta)^2}{x^2}S_{n+1}S_{n-1}+(1-\frac{(x^2+1+\beta)^2}{x^2})S_n^2.
\end{equation*}
This can be easily confirmed by employing the three term recurrence.

\medskip
We conclude this subsection by sketch the idea of the alternative proof of Theorem \ref{theorem_A1Q}.
It turns out to be a long journey to prove by
Sulanke and Xin's method and new transformations are needed. The proof is too lengthy to be typed here
but is easy to perform using Maple.

The generating function $F(t)=p_0+p_1 t+ \cdots$ is uniquely determined by
$$ F(t)=\frac{ tb-x}{ -1+{\frac { \left( b+1-{x}^{2} \right) }
{x}}t+tF \left( t \right)  }.$$
This simple form has to be transformed so that Lemma \ref{theorem_xin} applies. Indeed, by using Proposition 4.2 in \cite{sulanke2008hankel}, we obtain
$$ G(t)= \frac{ {\frac {b+{x}^{2}}{{x}^{2}}}-{\frac {b}{{x}^{3}}}t}
{ 1-{\frac { \left( b+1+{x}^{2} \right) }{x}}t+2\,{\frac {b}{{x}^{2}}}{t}^{2
}+{t}^{2} \left( -1+{\frac {b}{x}}t \right) G(t)
},$$
with the relation $H_n(F(t))=x^n H_{n-1}(G(t))$. The recursion \eqref{new_recurrence} we need to show becomes
\begin{align}
  T_nT_{n-2} -T_{n-1}^2 -\beta x^{-2n}=0,\label{T-recurrence}
\end{align}
where $T_n=H_n(G(t))$.

Now we prove by induction the recursion \eqref{T-recurrence} with $n$ replaced by $n+3$.
Eliminate $T_{n+3}$ by the following recursion
$${\frac {T_{{n+3}}T_{{n}}}{T_{{n+2}}T_{{n+1}}}}+{\frac {T_{{n+2}}T_{{n-
1}}}{T_{{n+1}}T_{{n}}}}={\frac {{b}^{2}+2\,b+2\,b{x}^{2}+4\,{x}^{2}+1+
{x}^{4}}{{x}^{2}}}-{\frac { \left( b+1+{x}^{2} \right) ^{2}{T_{{n+1}}}
^{2}}{{x}^{2}T_{{n+2}}T_{{n}}}}$$
which is obtained by application of Lemma \ref{theorem_xin}.

Successively eliminate $T_{n+2},T_{n+1}$ by recursion \eqref{T-recurrence} with $n$ replaced by $n+2$, $n+1$.
We are left to show the nullity of a polynomial, say $target(T_n,T_{n-1})$, of degree $8$ in $T_n$ and degree $4$ in $T_{n-1}$.
Now eliminate $T_{n}$ as before, we obtain a new polynomial which has $target(T_{n-1},T_{n-2})$ as a factor.
The proof is then completed after checking $target(T_1,T_0)=0$.

\section{Some Somos-polynomials}


Our starting point is to try to give a proof of the following statement of Somos in \cite{website:somospoly} by using the above result.
\begin{theorem}[Somos-4 Polynomials]\label{theorem-somos4polynomial}
The sequences produced by
   \begin{equation}\label{somos4poly}
    S_n=\frac{xyz S_{n-1}S_{n-3}+xyz S_{n-2}^2}{S_{n-4}}
 \end{equation}
with $S_0=x$, $S_1=1$, $S_2=1$, $S_3=y$ are all polynomials in $x,y,z$.
\end{theorem}

We obtain the following corollary as a consequence of Theorem \ref{theorem_somos4} by making the appropriate substitutions and calculations.
\begin{corollary}\label{somos4poly_det}
  If we let $S_{-2}=x$,$S_{-1}=1$,$S_0=1$,$S_1=y$,$S_2=y^2z+yz$,$S_3=xy^3z^2+xy^3z+xy^2z^2$ by shifting the indices, then any Somos-4 polynomial can be expressed as
\[
S_n=\det(\tilde{p}_{i+j})_{0\leq i,j\leq n-1},
\]
where $\tilde{p}_m=\frac{-yz+xy-z-xz}{\sqrt{xyz}}\tilde{p}_{m-1}+(x-y)\tilde{p}_{m-2}+\sum_{k=0}^{m-2}\tilde{p}_k\tilde{p}_{m-2-k}(m\geq2)$ with initial values $\tilde{p}_0=y,\tilde{p}_1=-\sqrt{xyz}.$
\end{corollary}

Our first hope was to obtain nicer determinant with polynomial entries by some transformations, but we find a simple proof by just using the Laurent property.

\begin{proof}[Proof of Theorem \ref{theorem-somos4polynomial}]
By the Laurent property, $S_n$ are Laurent polynomial with denominator factors $x$ and $y$. To see that $x$ and $y$ disappear in the denominator, we  extend the recursion \eqref{somos4poly} and obtain
$$ S_{-1}=xz(x+1), S_0=x, S_{1}=1, S_2=1, S_3=y, S_4=yz(y+1).$$
Using the Laurent property for $\{S_n\}_{n\geq -1}$, we see that the factor $y$ is not in the denominator;
Using the Laurent property for $\{S_n\}_{n\geq 1}$, we see that the factor $x$ is not in the denominator.
This completes the proof.
\end{proof}

The idea of the proof suggests the existence of a stronger version of the Laurent property. Indeed such a version was proposed by Hone and Swart \cite[Theorem 3.1]{hone2008integrality}.
\begin{theorem}\label{theorem_hone}
  For the Somos-4 recurrence \eqref{somos4}, every $S_n$ belongs to $\mathbb{Z}[\alpha, \beta, \alpha^2+\beta T, S_0^{\pm1},S_1,S_2,S_3]$, where $T$ is a quantity independent of $n$ defined by
  \begin{equation}\label{quantity_T}
    T=\frac{S_{n}S_{n+3}}{S_{n+1}S_{n+2}}+\alpha(\frac{S_{n+1}^2}{S_{n}S_{n+2}}+\frac{S_{n+2}^2}{S_{n+1}S_{n+3}})+\beta\frac{S_{n+1}S_{n+2}}{S_{n}S_{n+3}}.
  \end{equation}
\end{theorem}

With this theorem, we can obtain more general polynomial sequences.
\begin{theorem}\label{theorem_somos4poly}
  The sequences produced by
  \begin{equation}
    S_n=\frac{xyzw S_{n-1}S_{n-3}+xyzw S_{n-2}^2}{S_{n-4}}
 \end{equation}
   with the initial values for the following cases all yield polynomials in $x,y,z,w$:
  (1): 1,x,w,y; (2): x,1,w,y; (3): x,w,1,y; (4): x,w,y,1.
\end{theorem}
\begin{proof}
  Noting that $\beta=xyzw$ and $T$ is independent of $n$, from the expression of $T$, we conclude that $\beta T$ must be polynomials for all the four cases. And we always shift the indices to make $S_0=1$ without changing the value of $T$. Thus, the result follows by Theorem \ref{theorem_hone}.
\end{proof}

{\bf Remark:} By setting $w=1$, the special case (2) of the above theorem reduces to Theorem \ref{theorem-somos4polynomial}. In this case
as in Corollary \ref{somos4poly_det}, we have $\alpha=\beta=xyzw$, $S_0=1$ and $T=yz+xy+xz+z$.

Combining the idea of our proof of Theorem \ref{theorem-somos4polynomial} and using Theorem \ref{theorem_hone}, we can obtain more general result.
\begin{theorem}\label{theorem_somos4poly_general}
    Let $\{x_i\}_{i=1}^N$ be a finite set of indeterminates and $T$ be defined by \eqref{quantity_T}. Suppose $\{S_n\}$ satisfy the recursion system \eqref{somos4} and $\alpha,\beta,$ $\alpha^2+\beta T$ are polynomials in all  $x_i$. If there exists a nonnegative integer $r$ such that $\gcd(S_0,S_1,\dots, S_r)=1$ and $S_k$ are polynomials in all $x_i$ for $0\le k\le r+3$, then $S_n$ for $n\geq0$ are all polynomials in all $x_i$.
\end{theorem}
\begin{proof}
We prove it by contradiction.

Suppose that $D$ is an irreducible polynomial appearing as a factor in the denominator of certain $S_m$. Apply Theorem \ref{theorem_hone} to the Somos-4 sequences $\{S_n\}_{n\ge \ell}$ for $\ell=0,\dots, r$. We see that $D$ must be a factor of $S_\ell$ for each $\ell$. Thus $D$ divides $\gcd(S_0,\dots, S_r)=1$. A contradiction.
\end{proof}

\medskip
Hone and Swart also gave a strong Laurent property for Somos-5 \cite[Theorem 3.7]{hone2008integrality}.
\begin{theorem} \label{theorem_hone_somos5}
  For the Somos-5 recurrence \eqref{somos5}, any $\{S_n\}\in$  $\mathbb{Z}[\tilde\alpha, \tilde\beta, \tilde\beta+\tilde\alpha\tilde{T}, $ $S_0^{\pm1}, S_1^{\pm1}, S_2, S_3, S_4]$, where $\tilde{T}$ is a quantity independent of $n$ defined by
  \begin{equation}\label{quantity_tildeT}
    \tilde{T}=\frac{S_{n}S_{n+3}}{S_{n+1}S_{n+2}}+\frac{S_{n+1}S_{n+4}}{S_{n+2}S_{n+3}}+\tilde\alpha(\frac{S_{n+1}S_{n+2}}{S_{n}S_{n+3}}+\frac{S_{n+2}S_{n+3}}{S_{n+1}S_{n+4}})+\tilde\beta\frac{S_{n+2}^2}{S_{n}S_{n+4}}.
  \end{equation}
\end{theorem}

Consider the recurrence
   \begin{equation}\label{somos5poly}
    S_n=\frac{wxyz S_{n-1}S_{n-4}+wxyz S_{n-2}S_{n-3}}{S_{n-5}}.
 \end{equation}
Similar to the proof of Theorem \ref{theorem_somos4poly}, we can also confirm the following result for Somos-5 polynomials.

\begin{theorem}\label{theorem_somos5poly2}
  The sequences produced by \eqref{somos5poly} with the initial values for the following cases all yield polynomials in $x,y,z,w$:
  (1): $1,x,w,1,y$; (2): $x,1,w,y,1$; (3): $1,x,w,y,1$; (4): $1,x,1,w,y$; (5): $x,w,1,y,1$; (6): $x,1,w,1,y$; (7): $1,1,x,w,y$; (8): $x,1,1,w,y$; (9): $x,w,1,1,y$; (10): $x,w,y,1,1$.
\end{theorem}
\begin{proof}
Here we only give a detailed proof for case (1) because the proofs for the other cases are similar.

Noting that $\tilde\alpha=xyzw$ and $\tilde{T}$ is independent of $n$, from the expression of $\tilde{T}$, we conclude that $\tilde\alpha \tilde{T}$ must be polynomials for all the four cases. Thus, from Theorem \ref{theorem_hone_somos5} and $S_0=1$, $S_1=x$, it is obvious that $S_n$ is a Laurent polynomial in $x$ and polynomial in $y,z,w$. If we shift the indices to make $S_0=w,S_1=1$ without changing the value of $\tilde{T}$, then $S_n$ is also a Laurent polynomial in $w$ and polynomial in $x,y,z$. Therefore, $S_n$ must be a polynomial in $x,y,z,w$ and the proof is completed.
\end{proof}

Similar to Theorem \ref{theorem_somos4poly_general}, using  Theorem \ref{theorem_hone_somos5}, we also have the following general result for Somos-5 polynomials. Here the details of proof are omitted.
\begin{theorem}\label{theorem_somos5poly_general}
  Let $\{x_i\}_{i=1}^N$ be a finite set of indeterminates and $\tilde T$ be defined by \eqref{quantity_tildeT}. Suppose $\{S_n\}$ satisfy the recursion system \eqref{somos5} and $\tilde\alpha,\tilde\beta,$ $\tilde\beta+\tilde\alpha \tilde{T}$ are polynomials in all $x_i$. If there exists a positive integer $r$ such that $\gcd(S_0S_1,S_1S_2,$ $\dots, S_rS_{r+1})=1$ and $S_k$ are polynomials in all $x_i$ for $0\le k\le r+4$, then $S_n$ for $n\geq0$ are all polynomials in all $x_i$.
\end{theorem}

\section{Conclusion and discussions}

We have derived determinant solutions to three discrete integrable systems and confirm that the three discrete systems exhibit the Laurent property.
The simplest case is the extend $A_1$ $Q$-system. In particular, we prove that every Fibonacci number can be expressed as a Hankel determinant.
The general solution to the Somos-4 recurrence is given in terms of Hankel determinants. Observing that the Somos-5 recurrence can be viewed as a specified B\"{a}cklund transformation of the Somos-4 recurrence, we also derive the Hankel determinants solution to the Somos-5 recurrence.

As for the Laurent property of the Somos-4,5 recurrences, it can also be proved by using the combinatorial models as reduction cases of the octahedron recurrence \cite{speyer2007perfect} (also appearing as $T$-system in \cite{di2010solution,di2013t}). The $A_1$ $Q$-system, as a special case of the $Q$-system, was also done in \cite{di2010qsystem}. But the combinatorial method seems somewhat awkward to work with in practice. Our determinant formula consists of elements with convolution recurrence, and it is a closed form, which appears clearer and more intuitive. (Although the determinant solution to T system with arbitrary initial values was given in \cite{di2010solution,di2013t}, it is not a easy thing to construct the solutions to Somos-4,5 recurrences as reduction cases.) It should be noted that the combinatorial models imply the positive Laurent property (that is, the nonzero coefficients of the Laurent polynomials are all positive), while our determinant solutions do not seem to imply this property. The positive Laurent property is also an interesting topic. See \cite{di2010solution,di2011discrete,di2009positivity,di2010qsystem,di2010q,di2013t,fomin2000total,fomin2002cluster} etc.


Finally, we shall describe how we discovered the Hankel determinants formulae. We explain it by taking Theorem \ref{theorem_somos4} as an example. Actually, by applying Lemma \ref{lemma_xin} and \ref{theorem_xin}, we can solve the Somos-4 recurrence for Hankel determinant solution, and the expression of the solution is not unique. This is equivalent to solving for  $a_0,b_0,c_0,d_0,f_0$ (noting that $e_n=-1$).

Let us sketch the idea as follows. As before, shift the indices so that $S_{-1}=S_0=1$,$S_1=x,S_2=y,S_3=\alpha y+\beta x^2$. It is easy to see that $a_0=x$ and $a_1=y/x^2$.
Similar to the proof to Theorem \ref{theorem_somos4}, in order to make the Somos-4 recurrence hold, we need to equivalently make another formula  ($T(n)=0$) hold, which can be obtained by applying Lemma \ref{theorem_xin}. Then we need the induction step $T(n)-T(n-1)=0$ and $T(2)=0$, from which there needs to make sure that
\begin{equation}\label{f_1recurrence}
  f_1=-\frac{a_0f_0+a_0c_0-k\sqrt\alpha}{a_0}, \quad k=\pm1
\end{equation}
and
\begin{eqnarray}
   &&a_0  ^{2}  a_1   ^{2}+2\,  a_0  ^{2}a_1 f_0 f_1 +3\, a_0  ^{2}a_1 c_0f_1+2\, a_0   ^{2}a_1  f_1  ^{2}+ a_0  ^{2}a_1 cf_0 +a_0  ^{2}a_1 {c_0}^{2}- a_0  ^{3} f_0  ^{2}\nonumber\\
   &&-2\, a_0  ^{3}f_0 c_0-2\, a_0   ^{3}f_0 f_1 - a_0  ^{3}{c_0}^{2}-2\, a_0  ^{3}c_0f_1 - a_0  ^{3} f_1   ^{2}-\beta-\alpha\,a_1=0,\label{T(0)=0}
\end{eqnarray}
respectively.

Additionally, from the recurrence relations, we have
\begin{equation}\label{a_1recurrence}
  a_{1}=-\frac{a_0^3e_0+a_0^2d_0-a_0b_0c_0+b_0^2}{a_0^2}.
\end{equation}

Noting that $f_1=-\frac{b_0}{a_0}$ and $a_0,e_0,a_1$ is known, we can express $b_0,c_0,d_0$ in terms of a free number $f_0$ by using \eqref{f_1recurrence}, \eqref{T(0)=0} and \eqref{a_1recurrence}.
Thus, we have the following result.
\begin{theorem}
  If we let $S_{-1}=1$,$S_0=1$,$S_1=x,S_2=y,S_3=\alpha y+\beta x^2$ by shifting the indices, then the Somos-4 recurrence \eqref{somos4} has the following explicit Hankel determinant representation:
\[
S_n=\det(p_{i+j})_{0\leq i,j\leq n-1}.
\]
Here sequence $\{p_n\}$ has the generating function
\begin{equation*}
  P(t)=\frac{a_{0}+b_{0}t}{1+c_{0}t+d_{0}t^2+t^2(e_{0}+f_{0}t)P(t)},
\end{equation*}
with
 \begin{eqnarray*}
   &&a_{0}=x,\\
   &&b_{0}=-{\frac {yf_0  \sqrt{\alpha}x+\beta\,{x}^{2}-{y}^{2}-\alpha\,y+{x}^{3}\alpha+k\alpha\,y}{y \sqrt{\alpha}}},\\
   &&c_{0}=-{\frac {\beta\,{x}^{2}-{y}^{2}-\alpha\,y+{x}^{3}\alpha+2\,yf_0  \sqrt{\alpha}x}{xy \sqrt{\alpha}}},\\
   &&d_{0}=-\frac {1}{{{x}^{2}y \sqrt{\alpha}}}[-\beta\,{x}^{3}f_0 +\beta\,{x}^{2}k \sqrt{\alpha}+\alpha\,yf_0 x-{x}^{4}\alpha\,f_0 +{x}^{3}{\alpha}^{3/2}k+{y}^{2}f_0 x\\
   &&\quad \quad -{\alpha}^{3/2}yk-y  f_0^{2} \sqrt{\alpha}{x}^{2}+{k}^{2}{\alpha}^{3/2}y-{x}^{3}y \sqrt{\alpha}+{y}^{2} \sqrt{\alpha}-{y}^{2}k \sqrt{\alpha}],\\
   &&e_{0}=-1
 \end{eqnarray*}
and $k=\pm1$ and arbitrary number $f_0$.
  \end{theorem}

{\bf Remark:} The special case of $k=1$ and $f_0=0$ reduces to Theorem \ref{theorem_somos4}, which seems to be the most concise version.

Additionally, it is noted that one can find artificial interpretations in terms of weighted sums of some combinatorial objects for the general Somos-4,5 and the extended $A_1$ $Q$-system by employing the well-known Gessel-Viennot-Lindstr\"{o}m theorem (cf. e.g.
\cite{bressoud1999proofs,gessel1985binomial,sulanke2008hankel}) for determinants.

\section*{Acknowledgments}
We are grateful to the anonymous referees for helpful suggestions. This work was partially supported by the National
Natural Science Foundation of China (Grant Nos. 11331008, 11371251) and the
knowledge innovation program of LSEC and the Institute of
Computational Mathematics, AMSS, CAS. The third named author was partially supported by the National
Natural Science Foundation of China (Grant No. 11171231).


\end{document}